\documentclass[10pt,conference,a4paper]{IEEEtran}

\usepackage{amsmath}
\usepackage{amsfonts}
\usepackage{amssymb}
\usepackage{tabularx}{\tiny }
\usepackage{upref}
\usepackage{theorem}

\usepackage{graphics}
\usepackage[pdf_tex]{graphicx} 

\usepackage{psfrag}
\usepackage{enumerate}
\usepackage{caption}
\usepackage{subcaption}
\usepackage{color}
\usepackage{booktabs}
\usepackage{multirow}
\usepackage{multicol}
\usepackage{diagbox}
\usepackage{pict2e}
\usepackage{longtable}
\usepackage{bigdelim}
\usepackage[font=small,labelfont=bf,tableposition=top]{caption}

\newtheorem{theorem}{Theorem}
\newtheorem{lemma}{Lemma}

\newtheorem{corollary}{Corollary}
\newtheorem{definition}{Definition}
\newtheorem{remark}{Remark}
\newtheorem{example}{Example}

\definecolor{CommentBlue}{RGB}{0,80,239}
\definecolor{CommentPurple}{RGB}{145, 32, 186}
\definecolor{CommentRed}{RGB}{168,24,24}
\definecolor{CommentGreen}{RGB}{10,160,10}
\definecolor{CommentOrange}{RGB}{204,102,0}

\newcommand{\comment}[1]{}

\def\cG{\mbox{$\cal{G}$}}

\def\cS{\mbox{$\cal{S}$}}

\def\cD{\mbox{$\cal{D}$}}
\def\cJ{\mbox{$\cal{J}$}}

\def\cP{\mbox{$\cal{P}$}}

\def\cC{\mbox{$\cal{C}$}}
\def\DRS{\mbox{$\cal{D}_{RS}$}}
\def\bd{\mbox{$\bar{d}$}}

\newcommand{\TRS}{X{}} 

\newcommand{\ND}[1]{{\bar{D}_{\tilde{#1}}}}

\newif\ifncc
\newif\ifarxiv
\arxivtrue 

\title{Demand-Private Coded Caching and the Exact Trade-off for N=K=2}
\author{
	\IEEEauthorblockN{Sneha~Kamath\IEEEauthorrefmark{1}, Jithin~Ravi\IEEEauthorrefmark{2}, Bikash~Kumar~Dey\IEEEauthorrefmark{3}}
	\IEEEauthorblockA{\IEEEauthorrefmark{1} Qualcomm, India. Email: snehkama@qti.qualcomm.com}
	\IEEEauthorblockA{\IEEEauthorrefmark{2} Universidad Carlos III de Madrid, Legan\'es, Spain. Email: rjithin@tsc.uc3m.es}
	\IEEEauthorblockA{\IEEEauthorrefmark{3} Indian Institute of Technology Bombay, Mumbai.
		Email: bikash@ee.iitb.ac.in}
\thanks{J.~Ravi  received funding from the European Research Council (ERC) under the European Union's Horizon 2020 research and innovation programme (Grant No.~714161). The work of
	B.~K.~Dey was supported in part by the Bharti Centre for Communication, lIT Bombay.}	
 }	
		
\begin{document}
  	\ifarxiv
  \IEEEoverridecommandlockouts
 \maketitle
 \fi

 \begin{abstract} 
The distributed coded caching problem has been studied extensively in
the recent past.
While the known coded caching schemes achieve an improved transmission
rate, they violate the privacy of the users since in these schemes the demand
of one user is revealed to others in the delivery phase. In this paper, we
consider the coded caching problem under the constraint that
the demands of the other users remain information theoretically secret
from each user.
We first show that the memory-rate pair $(M,\min \{N,K\}(1-M/N))$ is achievable
under information theoretic demand privacy, while using broadcast transmissions. Using this, we show that perfectly demand-private coded caching rate
is order optimal for all parameter regimes. We  then show that a demand-private scheme for $N$ files and $K$ users
can be obtained from a non-private scheme that satisfies only a restricted
subset of demands of $NK$ users for $N$ files. We then focus on the demand-private coded caching problem for $K=2$ users, $N=2$ files. We characterize the
exact memory-rate trade-off for this case. To show the achievability, we use 
our first result to construct a demand-private scheme from a non-private
scheme satisfying a restricted demand subset that is known from an earlier work
by Tian. Further, by giving a converse based on the extra requirement of
privacy, we show that the obtained achievable region is the exact
memory-rate trade-off.
 \end{abstract}


 \section{Introduction}
In the seminal work~\cite{Maddah14}, Maddah-Ali and Niesen demonstrated
that significant gain in the transmission rate can be achieved in a noiseless
broadcast network by clever design of  caching and delivery schemes. The
network studied in~\cite{Maddah14} consists of one server and $K$ users, each
user is equipped with a cache of uniform size. The server has $N$ files and
each user requests one of the $N$ files in the delivery phase. By utilizing the
broadcasting opportunity of this network, Maddah-Ali and Niesen provided a
caching and delivery scheme which is shown to be order optimal within a factor
of 12.

In this paper, we consider the coded caching problem under privacy requirement
on the demands of the users, i.e., no user should learn anything about the
demands of the other users.
Recently, demand privacy for the coded caching setup has been studied from an
information theoretic perspective~\cite{Wan19, Kamath19,
AravindST19}. In \cite{Wan19}, it was studied under
a setup where the delivery phase uses private multicasts to subsets of users, equivalently studying computational privacy guarantee (see Remark~\ref{rem:perfectprivacy}).
Coded caching under perfect information theoretic privacy was studied first 
in~\cite{Kamath19}.
In both \cite{Wan19, Kamath19}, construction techniques were proposed for
deriving a demand-private scheme for $N$ files and $K$ users from a
non-private coded caching scheme for $N$ files and $NK$ users. 
The achievable memory-rate pairs of the derived schemes are the same
as that of the original non-private schemes.
In~\cite{AravindST19}, authors study the subpacketization requirement under
information theoretic demand privacy constraint for $N=K=2$. They have shown some lower bounds on the
transmission rate for a given subpacketization when the caching scheme is
constrained to be linear.

The non-private coded caching problem has been studied
by many authors.
The works~\cite{Amiri17, Zhang18, Vilardebo18, Yu18} focused on improving the
achievable rates of Maddah-Ali and Niesen~\cite{Maddah14} by designing new schemes. Yu {\em et al.}~\cite{Yu18} proposed
a new caching scheme which was shown to be order optimal within a factor of 2.
When the cache content is not allowed to be coded, the optimal rates were
characterized in~\cite{Yu18, Wan16}.  Several works have obtained improved
lower bounds on the rates, see for example~\cite{Ghasemi17,Wang18}.

The coded caching schemes in the noiseless broadcast network is inherently
prone to security and privacy issues since the broadcasted message is revealed
to everyone. Information theoretic secrecy from an external adversary who can
observe the broadcasted message was first studied by Sengupta {\em et
al.}~\cite{Sengupta15}. They proposed a scheme which prevents the adversary
from getting any information about any file from the broadcasted message.
Another privacy aspect was considered by Ravindrakumar {\em et al.}
in~\cite{Ravindrakumar18} where each user should
not get any information about any file other than the one requested by her. They
proposed a scheme which achieves this constraint  by distributing keys in the
placement phase.

The contributions of this paper are as listed below.

\comment{
Recently, an idependent and parallel work~\cite{Wan19} on demand privacy for coded caching was posted on arXiv on 28 August 2019. The problem formulation of~\cite{Wan19} is very similar to ours.
 We studied the problem of demand privacy for coded caching for the case of single request from users, the setup studied in~\cite{Maddah14}.
The authors in~\cite{Wan19} studied the cases of single request as well as multiple requests from the users. So it is good to compare our results with theirs for the single request case.
For the single request case, the achievable schemes in these two works are very different. 
In both of these works, the  tightness of the achievable rates are shown  by comparing it with the existing converse results on coded caching problem under no privacy requirement. 
The general scheme in~\cite{Wan19} is shown to be order optimal when $M\geq N/2$. 
In contrast, we show that our scheme is order optimal within a factor of 8 when $N \leq K$.
A special case where all users demand distinct files was also studied in~\cite{Wan19}. This assumption implies that $N\geq K$. They have provided an improved scheme for this special case which is order optimal within a factor of $4$ when $M\geq N/K$. We show that when $N>  K$, our scheme is order optimal within a factor of $4$ when $M \geq N/K$ without any restriction on the demand vectors.
}

\begin{enumerate}
\item  We first show in Theorem~\ref{th:basic} that the memory-rate pair $(M,\min \{N,K\}(1-M/N))$
is achievable for coded caching under information theoretic demand
privacy. Our achievable scheme uses broadcast transmissions in the delivery
phase, and this complements a similar result in ~\cite{Wan19} for their
model using private unicast transmissions in the delivery stage. We conclude in
Theorem~\ref{Thm_order} that the optimal rates with and without
demand privacy are always within a multiplicative factor, and this completes the
order optimality~\cite{Kamath19} of information theoretically demand-private coded caching in all memory regimes.

	\item  We show in Theorem~\ref{Thm_genach} that a demand-private scheme for $N$ files and $K$
users with the same memory-rate pair $(M,R)$ can be obtained from a non-private
scheme that serves only a subset of demands for $N$ files and $NK$ users.
This is a refinement of results of ~\cite{Wan19,Kamath19}, and the scheme uses
the idea in~\cite{Kamath19}.
However, the observation that the
particular non-private scheme is required to serve only a subset of demands is
new, and this is used later for the case of $N=K=2$, discussed
in the next item.
\item In Theorem~\ref{Thm_N2K2_region}, we characterize the exact
memory-rate trade-off with demand privacy for $N=K=2$.
We note that the region given in
Theorem~\ref{Thm_N2K2_region}  is strictly larger than achievable
regions known from existing literature (See Fig.~\ref{Fig_region_N2K2}). To obtain this
achievable region, we use two non-private caching schemes from~\cite{Tian2018}
which are required to serve  a restricted subset of demands.
Proving converse for this problem is difficult in general, and the converse
proof of Theorem~\ref{Thm_N2K2_region} is a key contribution of this
paper.
	
	\item The achievability of the exact memory-rate trade-off in
Theorem~\ref{Thm_N2K2_region} is proved by showing that memory-rate pairs
$(1/3,4/3)$ and $(4/3,1/3)$ are achievable. The caching and transmission
schemes to achieve these points are linear with coded prefetching. Incidentally,
these schemes also use a subpacketization of 3, which is the same as that
of the schemes in~\cite{AravindST19} for the rate points $(2/3,1)$ 
and $(1,2/3)$.
The question of whether the minimum required subpacketization is indeed 3 to achieve any
memory-rate pair with demand privacy for $N=K=2$ remains open.
\end{enumerate}

We present the problem formulation and definitions in Sec.~\ref{sec_problem}.
The results with proofs are presented in Sec.~\ref{sec_results}.


%
 \section{Problem formulation and definitions}
 \label{sec_problem}
 
Consider a server with $N$ files $W_0, W_1,\ldots,W_{N-1}$ which are assumed to
be independent and each of length $F$ bits. File $W_i, i=0, \ldots,
N-1$ takes values  in the set   $[2^F] := \{0,1,\ldots, 2^F-1\}$ uniformly at
random.  The server is connected to $K$ users  via a noiseless broadcast link.
Each user is equipped with a cache of size $MF$ bits, where  $M \in [0,N]$.
There are two phases in a coded caching scheme. In the first phase, called the
placement phase, the server fills the cache of each user.  In the delivery
phase, each user requests one file from the server. The index of the file 
requested by user $k$ is denoted by $D_k$. We assume that all $D_k$ are
independent of each other, and each of them is uniformly distributed in the set
$[N]$. Let  the vector $\bar{D}=[D_{0},D_{1},...,D_{K-1}]$ denote the demands
of all users, and also let  $\ND{k}$ denote all demands but $D_k$, i.e.,
$\ND{k}= \bar {D} \setminus \{ D_{k}\}$.  All users convey their demands
secretly to the server. Then, the server broadcasts a message of size $RF$ bits
to serve the request of the users. The broadcasted message, denoted by $X$,
consists of $RF$ bits, where $R$ is defined as the rate of transmission. User
$k$ decodes the requested file $W_{D_k}$ using  the received message,
cache content, and $D_k$.

In a demand-private coded caching setup (see Fig.~\ref{Fig_cach_setup}), we
also have a privacy requirement on the demand in addition to the recovery
requirement.  The privacy constraint is such that user $k$ should not gain any
information about $\ND{k}$. 
To achieve this, we consider some {\em shared randomness} $S_k$ which is shared between user $k$ and the server, and  it is not known to the other users.
The shared randomness can be achieved during the placement phase since the placement is done secretly for each user.
 Random variables $S_0,\ldots, S_{K-1}$ take values in some finite alphabets $\cS_0,\ldots, \cS_{K-1}$, respectively.
The set of random variables $(S_0, \ldots, S_{K-1})$ is denoted by $\bar{S}$. Let $\cP$ denote the set of values of  a private randomness $P$ available at the server.
The random variables $P \cup \{S_k:  k \in [K] \} \cup  \{D_k: k \in [K] \} \cup \{W_i: i \in [N]\}$  are independent of each other.

\begin{figure}[htb]
  \centering
   \includegraphics[scale=0.4]{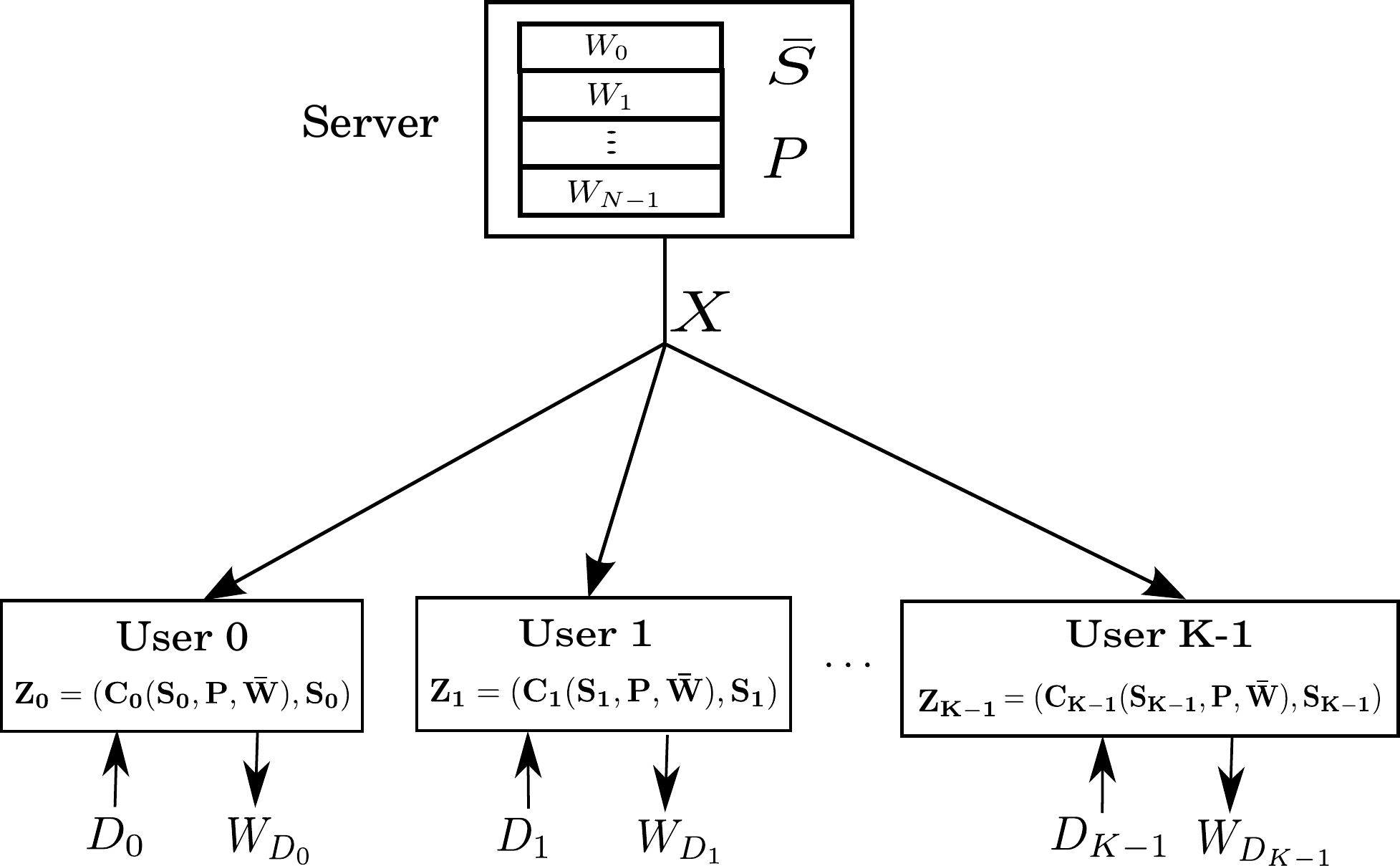}
  \caption{Demand-private coded caching model.}
  \label{Fig_cach_setup}
\end{figure}

{\bf Non-private coded caching scheme:}
A {\it non-private coded caching scheme} consists of the following.

{\it Cache encoding functions:} For $k \in [K]$, the cache encoding function for the $k$-th user is a map
\begin{align}
C_{k}: {[2^F]}^N  \rightarrow [2^{MF}], \label{Def_cach_enc_np}
\end{align} 
and the cache content $Z_k$ is given by $Z_k =C_k(\bar{W})$. 

{\it Broadcast transmission encoding function:} The  transmission
encoding is a map
\begin{align}
E: {[2^F]}^N \times \cD_0  \times \cdots \times \cD_{K-1}  \rightarrow [2^{RF}], \label{Def_Tx_enc_np}
\end{align}
and the transmitted message is given by $X=(E(\bar{W}, \bar{D}), \bar{D})$. 

{\it Decoding functions:} User $k$ uses a decoding function
\begin{align}
G_{k}:  \cD_0 \times \cdots \times \cD_{K-1}  \times [2^{RF}] \times [2^{MF}]  \rightarrow [2^{F}]. \label{Def_dec_np}
\end{align}
Let $\cC = \{C_k: k=0,\ldots,K-1\}$ and $\cG =  \{G_k: k=0,\ldots,K-1\}$. Then
the triple $(\cC, E, \cG)$ is called an
$(N,K,M,R)$-non-private scheme if it satisfies 
\begin{align} 
W_{D_{k}} = G_k(\bar{D}, E(\bar{W}, \bar{D}), C_k(\bar{W}))
\label{Eq_dec_cond}
\end{align}
for all values of $\bar{D}, \bar{W}$.
A memory-rate pair $(M,R)$ is said to be {\em achievable} for the $(N,K)$ coded
caching problem if there exists an $(N,K,M,R)$-non-private scheme  for some
$F$.

{\bf Private coded caching scheme:} A {\it private coded caching scheme}
consists of the following. 

{\it Cache encoding functions:}
For $k \in [K]$, the cache encoding function for the $k$-th user is given by
\begin{align}
C_{k} :\cS_k \times \cP \times {[2^F]}^N  \rightarrow [2^{MF}], \label{Def_cach_enc}
\end{align} 
and the cache content $Z_k$ is given by 
\mbox{$Z_k =(C_k(S_k, P, \bar{W}), S_k)$}. 

{\it Broadcast transmission encoding function:}
The   transmission encoding functions are
\begin{align*}
&E: {[2^F]}^N \times \cD_0 \times \cdots \times \cD_{K-1} \times\cP \\
& \qquad \qquad \qquad \times \cS_0 \times \cdots \times \cS_{K-1} \rightarrow [2^{RF}],\\
&J:\cD_0 \times \cdots \times \cD_{K-1} \times \cP \times \cS_0 \times \cdots \times \cS_{K-1} \rightarrow \cJ.
\end{align*}
The transmitted message $X$ is given by
\begin{align*}
 X=\left(E(\bar{W}, \bar{D}, P, \bar{S}), J(\bar{D}, P, \bar{S}) \right). 
\end{align*}
Here $\log_2 |\cJ|$ is negligible\footnote{ The auxiliary
transmission $J$ essentially captures any additional transmission, that does not
contribute any rate, in addition to the main
payload. Such auxiliary transmissions of negligible rate are used even in non-private
schemes without being formally stated in most work. For example, the
scheme in~\cite{Maddah14} works only if the server additionally transmits the demand vector
in the delivery phase. We have chosen to
formally define such auxiliary transmission here.} compared to file size $F$.

{\it Decoding functions:}
User $k$ has a decoding function
\begin{align}
G_{k} : \cD_k \times \cS_k \times \cJ \times [2^{RF}] \times [2^{MF}]  \rightarrow [2^{F}]. \label{Def_dec}
\end{align}
Let $\cC = \{C_k: k=0,\ldots,K-1\}$ and $\cG =  \{G_k: k=0,\ldots,K-1\}$. The
tuple  $(\cC, E, J,\cG)$ is called as an $(N,K,M,R)$-private scheme if it
satisfies the following decoding and privacy conditions:
\begin{align} 
&W_{D_{k}} = G_k\bigl(D_k,S_k,J(\bar{D}, P, \bar{S}, ), \notag\\
& \qquad \qquad \qquad \qquad E(\bar{W}, \bar{D},P,\bar{S}), C_k(S_k, P,\bar{W})\bigr), \label{Eq_decod_cond}\\
& I\left(\ND{k};Z_k,\TRS, D_k \right)  = 0 \quad \text{ for }   k=0,\ldots,K-1. \label{Eq_instant_priv}
\end{align}
A memory-rate pair $(M,R)$ is said to be {\em achievable with demand privacy}
for the $(N,K)$ coded caching problem if there exists an $(N,K,M,R)$-private
scheme  for some  $F$.


The {\em memory-rate trade-off with demand privacy} is defined as 
\begin{align}
R^{*p}_{N,K}(M)&=\inf\{R: (M,R) \mbox{ is achievable with demand }  \notag \\
&  \hspace*{-2mm}\mbox{privacy for $(N,K)$ coded caching problem.} \} \label{Eq_opt_rate_priv}
\end{align}
The {\em  memory-rate trade-off } $R^{*}_{N,K}(M)$ for the non-private coded caching problem
is defined similarly.

\ifarxiv 
The following result is known from~\cite{Kamath19,Wan19}.
\begin{theorem}\cite[Theorem~1]{Kamath19}
\label{th:known1}
 If there exists an $(N,NK,M,R)$ non-private scheme, then
there exists an $(N,K,M,R)$-private scheme.
\end{theorem}
\fi

\begin{remark}
\label{rem:perfectprivacy}
	The model studied in~\cite{Wan19} assumed that the server can 
privately transmit to any subset of users by encryption using shared keys. The key length required for achieving such private multicast under information-theoretic
privacy using broadcast transmissions is the same as the length of the
multicast message. In that case, storing such keys in the cache will
also contribute to the cache memory requirement. The required key rates are
negligible only under computational privacy requirement, as noted  in~\cite{Wan19}. Since
it was assumed that the shared keys are of negligible rates, the overall model
in~\cite{Wan19} does not ensure information-theoretic privacy under broadcast
transmission. In contrast, we assume broadcast transmission in the delivery phase and 
we study perfect privacy in information-theoretic sense.
 \end{remark}

 \section{Results}
 \label{sec_results}


In~\cite[Example~1]{Maddah14}, it was shown that we can achieve rate
$\min\{N,K\}(1-M/N)$ for non-private scheme without any coding in cache
placement or in broadcast transmission. Next we show that the same rate is
achievable under perfect privacy of the demands under broadcast transmissions. The achievability of this rate using private unicast transmissions
is simple \cite[Theorem 1]{Wan19}, and this implies the
achievability under computational privacy guarantee using broadcast transmissions.

\begin{theorem}
\label{th:basic}
For any $M$, the memory-rate pair $(M,\min \{N,K\}(1-M/N))$ is achievable in
coded caching under information theoretic demand-privacy.  
\end{theorem}

\begin{proof}
\ifarxiv
In the placement phase, the caches of all
users are populated with the same $M/N$ fraction of each file. 
\fi
Let each file $W_i$ be split in two parts: cached part $W^{(c)}_i$
of length $FM/N$, and uncached part $W^{(u)}_i$
of length $F(1-M/N)$.
The cache contents of all the users are the same, and given by
$Z_k=(Z^{(0)},Z^{(1)},\cdots,Z^{(N-1)})$, where
\ifarxiv
\begin{align*}
Z^{(i)}=W^{(c)}_i \text{ for }i=0,1,\cdots,N-1.
\end{align*}
\fi
To describe the delivery phase, we consider two cases:

\ifarxiv
\noindent
\underline{Case 1: $N \leq K$}
\fi

For $N<K$, the server broadcasts the remaining $(1-M/N)$ fraction of each
file. 
\ifarxiv
This scheme achieves privacy because the transmission does not depend
on the demands of the users.

\noindent
\underline{Case 2: $N > K$}

\fi
%
Let $D_0,D_1,\cdots,D_{K-1}$ be the demands of the users.
\ifarxiv
The transmission $X$ has two parts $(X',J)$, where 
$X'=(X'_0,X'_1,\cdots,X'_{K-1})$ is the main payload, and $J$ is
the auxiliary transmission of negligible rate which helps each user find the corresponding decoding function. 
For each $i$, $X'_i$
is either $W^{(u)}_{D_j}$ for some $j$ or random bits
of the same length. In particular,
the position of $W^{(u)}_{D_j}$ in $X'$ is denoted by a random variable
$P_j\in [K]$. 
\fi
The random variables $P_0,P_1,\cdots,P_{K-1}$ are defined inductively
as
\begin{align*}
P_i = & \begin{cases} 
 P_j & \text{if } D_i=D_j \\ & \text{for some } j<i \\
\sim unif([K]\setminus\{P_0,P_1,\cdots,P_{i-1}\})
& \text{if } D_i\neq D_j\\ & \forall j<i.
\end{cases}
\end{align*}
\ifarxiv
Note that each demanded (uncached) file is transmitted only in one component
of the transmission so that one user can not possibly detect the same file
(as its own demand) being transmitted in another component and thus infer
that the corresponding other user also has the same demand.
\fi

The keys $S_0,S_1,\cdots,S_{K-1}\in [K]$ are chosen i.i.d. and uniformly
distributed. 
The transmission is then given by
\begin{align*}
X'_{j} = \begin{cases} W^{(u)}_{D_i} & \text{if } j=P_i \\&  \text{for some }i\in [K]\\
  \sim unif\left( \{0,1\}^{F(1-M/N)}\right) & \text{otherwise.}
\end{cases}
\end{align*}
and
\ifarxiv
\begin{align*}
J = (P_0\oplus_K S_0, P_1\oplus_K S_1,\cdots,P_{K-1}\oplus_K S_{K-1}),
\end{align*}
\fi
where $\oplus_K$ denotes the addition modulo $K$ operation.
Since user $k$ knows $S_k$, it can find $P_k$ from $J$. It
then can find $X'_{P_k}=W^{(u)}_{D_k}$, and thus $W_{D_k}=(Z^{(D_k)},X'_{P_k})$.

Next we show that this scheme also satisfies the privacy condition.
Let us denote $Q_i=P_i\oplus_K S_i$ for the ease of writing.
\begin{align}
&I(\ND{k};X,  D_k,Z_k) \notag \\
&=I(\ND{k}; X'_0,\cdots,X'_{K-1}, Q_0,Q_1,\cdots,Q_{K-1}, \notag \\ 
& \qquad \qquad D_k, S_k, W^{(c)}_{0}, \ldots,W^{(c)}_{N-1} ) \notag \\
&=I(\ND{k};  Q_0, \cdots,Q_{K-1}, D_k, S_k)\label{eq:privacy1}\\
& = I(\ND{k};  Q_0,  \cdots,Q_{k-1},Q_{k+1},\cdots,Q_{K-1}, 
  D_k, S_k,P_k)\notag\\
&=0\label{eq:privacy3}
\end{align}
where~\eqref{eq:privacy1} follows because $(X'_0,\cdots,X'_{K-1},W^{(c)}_{0},
\ldots,W^{(c)}_{N-1} )$ is uniformly distributed in $\{0,1\}^{MF+FK(1-M/N)}$,
and is independent of $(\ND{k},Q_0, \cdots,Q_{K-1},D_k, S_k)$,
and \eqref{eq:privacy3} follows because all the random variables in the
mutual information are independent.
%
\ifarxiv
In this scheme, the number of bits broadcasted is $FK(1-M/N)$ 
as the bits transmitted for communicating $J$ is negligible for large $F$.
Thus, the scheme achieves  rate $K(1-M/N)$.
\fi
\end{proof}
	
One natural question that arises in demand-private coded caching is
how much cost it incurs due to the extra constraint of demand privacy. The next
theorem shows that the extra cost is always within a multiplicative factor of
8. 
\begin{theorem}
	\label{Thm_order}
The optimal rates with and without privacy always satisfy the following: 
\begin{align}
\frac{R^{*p}_{N,K}(M)}{R^{*}_{N,K}(M)} & \leq 8.
\end{align}
\end{theorem}
\begin{proof}
	The achievable memory-rate pair using the scheme given in~\cite{Kamath19} is
	shown~\cite[Theorem~2]{Kamath19} to be within a factor of 8 from
	${R^{*}_{N,K}(M)}$ for all memory regimes except for $0\leq M \leq  N/K$ when
	$N > K$. So, we have Theorem~\ref{Thm_order}  for all those memory
	regimes. Next we show that Theorem~\ref{Thm_order} also holds for $0\leq M \leq  N/K$ when $N > K$.
	
	First, let us consider $M=0$. Since $N > K$, it follows from Theorem~\ref{th:basic} that $R^{*p}_{N,K}(0) \leq K$. 
	So, it is clear that
		\begin{align}
	R^{*p}_{N,K}(0) \leq K \quad \text{ for } M\geq 0. \label{Eq_Rp_uppr}
	\end{align}
	Now let us consider the scheme given in~\cite{Yu18} under no privacy constraint and let $R^{\text{YMA}}_{N,K}(M)$ denote the achievable rate
	using this scheme. For $M=N/K$, $R^{\text{YMA}}_{N,K}(M)$ is given by
	\begin{align}
		R^{\text{YMA}}_{N,K}(N/K) = (K-1)/2. \notag
	\end{align}
	
Further, $	R^{\text{YMA}}_{N,K}(M)$ is monotonically non-increasing in $M$  which implies that 
\begin{align}
R^{\text{YMA}}_{N,K}(M) \geq (K-1)/2 \quad \text{ for } 0\leq M \leq  N/K. \label{Eq_YMA_lowr}
\end{align}
Then, it follows from~\eqref{Eq_Rp_uppr} and~\eqref{Eq_YMA_lowr} that, for $0\leq M \leq  N/K$  
\begin{align}
\frac{R^{*p}_{N,K}(M)}{	R^{\text{YMA}}_{N,K}(M)} &\leq \frac{2 K}{K-1} \notag \\
& \leq 4 \quad \text{ for } K \geq 2. \label{Eq_Rp_bnd}
\end{align}
Furthermore, it was shown in~\cite{Yu19} that
\begin{align}
\frac{	R^{\text{YMA}}_{N,K}(M)}{R^{*}_{N,K}(M )} & \leq 2. \label{Eq_R_YMA_bnd}
\end{align}
So, \eqref{Eq_Rp_bnd} and \eqref{Eq_R_YMA_bnd} imply that
\begin{align}
\frac{R^{*p}_{N,K}(M)}{R^{*}_{N,K}(M)} & \leq 8 \text{ for } 0\leq M \leq  N/K.
\end{align}
This completes the proof of the theorem. 
\end{proof}

Theorem~\ref{Thm_order} completes the order optimality result of demand-private coded
caching. We also note that under computational privacy guarantee,
the order optimality for all memory regimes was given in~\cite{Wan19}.

A demand-private scheme for $N$ files and $K$ users can be obtained using
an existing non-private achievable scheme for $N$ files and $NK$
users as a blackbox. Here every user is associated with a stack of $N$ users in
the non-private caching problem. For
example, demand-private schemes for $N=K=2$ are obtained from the non-private
schemes for $N=2$ and $K=4$. 
We use the ideas from the scheme presented in
\cite{Kamath19}, where only certain types of demand vectors
for the non-private scheme are used in the private scheme. 
Next we define this particular subset of demand vectors.

Consider a non-private coded caching problem with $N$ files and $NK$ users. 
A demand vector $\bd$ in this problem is an $NK$-length vector, 
where the $j^{\text{th}}$ component denotes the demand of user $j$.
	Then $\bd$ can also be represented as $K$ subvectors of length $N$, i.e.,	
	$$\bd= [\bd^{(0)},\bd^{(1)},\ldots,\bd^{(K-1)}],$$
where $\bd^{(i)}\in [N]^N$ is  an $N$-length vector for all $i \in [K]$. 
We now define a ``restricted demand subset'' $\DRS$.

 \begin{definition}[Restricted Demand Subset $\DRS$]
	\label{Def_dmnd_subst}
The restricted demand subset $\DRS$ for an $(N,NK)$ coded caching problem
is the set of all $\bd$ such that $\bd^{(i)}$ is a cyclic shift of the vector $(0, 1,  \ldots, N-1)$ for all $i=0,1, \ldots, K-1$.
\end{definition}
		Since $N$ cyclic shifts are possible for each $\bd^{(i)}$, there are a total of $N^K$ such demand vectors in $\DRS$. 

For a given $\bd \in \DRS$ and $i\in [K]$, let $c_i$ denote the number of right cyclic
shifts of $(0,1,\ldots,N-1)$ needed to get $\bd^{(i)}$. Then, 
$\bd \in \DRS$ is  uniquely identified by the vector $\bar{c}(\bd)
:= (c_1, \ldots, c_K)$. For $N=2$ and $NK =4$, the demands in $\DRS$ and
their corresponding $\bar{c}(\bar{d}_s)$ are given in Table~\ref{Tab_2x2}.

\begin{table}[h]
	\begin{center}
		\begin{tabular}{|c|c|c|c|c|}
			\hline
			$D_{0}$ & $D_{1}$ & $D_{2}$ & $D_{3}$ & $\bar{c}(\bar{d}_s)$ \\
			\hline
			$0$ & $1$ & $0$ & $1$ & $(0,0)$\\
			\hline
			$0$ & $1$ & $1$ & $0$ & $(0,1)$\\
			\hline
			$1$ & $0$ & $0$ & $1$ & $(1,0)$\\
			\hline
			$1$ & $0$ & $1$ & $0$ & $(1,1)$\\
			\hline
		\end{tabular}
	\end{center}
	\caption{Demand subset $\DRS$ for $N=2$ and $NK=4$.}
	\label{Tab_2x2}
\end{table}

A related concept is the ``demand type'' used in \cite{Tian2018}.
\begin{definition}[Demand Types]
	In $(N,K)$-non-private coded caching problem, for a given demand vector
$\bar{d}$, let $t_i$ denote the number of users requesting file $i$, where
$i=0, \ldots, N-1$. Demand type of $\bar{d}$, denoted by $T(\bar{d})$, is defined
as the $N$-length vector 
$T(\bar{d}):=\bar{t} = (t_1, \ldots, t_N)$. The type class of
$\bar{t}$ is defined as $\cD_{\bar{t}}=\{\bd|T(\bd)=\bar{t}\}$. \end{definition}

Clearly, the restricted demand subset $\DRS$ is a subset of the type class
$(K,K,\ldots,K)$, i.e., 
\begin{eqnarray}
&&\DRS \subseteq \cD_{(K,K,\ldots,K)}.\label{eq:drstype}
\end{eqnarray}

A non-private scheme for an $(N,K)$ coded caching problem that satisfies all
demand vectors in a particular demand subset $\cD\subset[N]^K$, is called a
$\cD$-non-private scheme. Clearly, for $\cD_1\subset \cD_2$, a
$\cD_2$-non-private scheme is also a $\cD_1$-non-private scheme. In particular,
achievable rates for satisfying various demand
type-classes were studied in \cite{Tian2018}, and their
results are useful in our schemes for the type $(K,K,\cdots,K)$ due to the
relation \eqref{eq:drstype}.

\ifarxiv
We now present a refinement of Theorem~\ref{th:known1}.
\fi

 \begin{theorem}
        \label{Thm_genach}
 If there exists an $(N,NK,M,R)$ \DRS-non-private scheme, then
there exists an $(N,K,M,R)$-private scheme.
 \end{theorem}

The proof will construct an $(N,K,M,R)$-private scheme using an $(N,NK,M,R)$
\DRS-non-private scheme as a blackbox using ideas from~\cite{Kamath19}. 
We first give an example to illustrate this construction for $N=2,K=2$
using only the restricted demand subset for a $(2,4,\frac{1}{3},
\frac{4}{3})$ $\cD_{(2,2)}$-non-private scheme from ~\cite{Tian2018}.  We will
see that this allows a better achievable rate $(\frac{1}{3}, \frac{4}{3})$ for
the $(N=2,K=2)$ demand-private
coded caching problem than what can be achieved for the $N=2,K=4$ non-private
caching problem. 

\begin{example}
	\label{Ex_simple}
We consider the demand-private coded caching problem for $N=2,K=2$.
Using results from \cite{Kamath19} and \cite{Wan19}, we know that a
demand-private scheme of the same rate-pair can be obtained from any 
non-private scheme for $N=2,K=4$. However, it was shown in \cite{Tian2018} that
for the memory $M=1/3$, the optimum transmission rate $R^{*}_{2,4}(1/3) > 4/3$.
It can be shown that other demand-private schemes in \cite{Wan19} also
do not achieve $R=4/3$ for $N=2,K=2$. See Fig.~\ref{Fig_region_N2K2} for reference.

Let $A$ and $B$ denote the two files. We will now
give a scheme which achieves a rate $4/3$ for $M=1/3$ with $F=3l$. We denote the
3 segments of $A$ and $B$ by $A_1,A_2,A_3$ and $B_1,B_2,B_3$ respectively, of $l$ bits each.  
First let us consider a $\DRS$-non-private scheme for $N=2$ and $K=4$ from
\cite{Tian2018}. Let
$C_{i,j}(A,B)$, as shown in Table~\ref{Table_cache_NK2}, correspond to the cache
content of user $2i+j$ in the $\DRS$-non-private scheme.  The transmission
$T_{(i,j)}(A,B), i,j=0,1$, as given in 
\ifarxiv Table~\ref{Tab_Tx}, \fi
is chosen for the demand
$\bar{d} \in \DRS$ such that $ (i,j) = \bar{c}(\bar{d})$. Using
\ifarxiv Tables~\ref{Table_cache_NK2} and \ref{Tab_Tx}, \fi 
it is easy to verify that the non-private scheme satisfies the decodability condition for demands in \DRS. From this
scheme, we obtain a demand-private scheme for $N=2,K=2$ as follows. 
\ifarxiv

\begin{table}[h]
        \centering
        \begin{tabular}{|c|c|}
                \hline
                Cache & Cache Content\\
                \hline
                $C_{0,0}(A,B)$ & $A_1\oplus B_1$ \\
                \hline
                $C_{0,1}(A,B)$ & $A_3\oplus B_3$\\
                \hline
                $C_{1,0}(A,B)$ & $A_2\oplus B_2$ \\
                \hline
                $C_{1,1}(A,B)$ & $A_1\oplus A_2\oplus A_3\oplus B_1\oplus B_2\oplus B_3$\\
                \hline
        \end{tabular}
        \caption{Choices for the caches of user 1 and 2.}
        \label{Table_cache_NK2}
\end{table}

\begin{table}[h]
	\centering
	\begin{tabular}{|c|c|}
		\hline
		$T_{(0,0)}(A,B)$ &  $B_1, B_2, A_3, A_1\oplus A_2\oplus A_3$ \\
		\hline
		$T_{(0,1)}(A,B)$ &  $A_2, A_3, B_1, B_1\oplus B_2\oplus B_3$ \\
		\hline
		$T_{(1,0)}(A,B)$ &  $B_2, B_3, A_1, A_1\oplus A_2\oplus A_3$ \\
		\hline
		$T_{(1,1)}(A,B)$ &  $A_1, A_2, B_3, B_1 \oplus B_2 \oplus B_3 $ \\
		\hline
	\end{tabular}
	\caption{Transmissions for $(2,2, \frac{1}{3}, \frac{4}{3})$-private scheme.}
	\label{Tab_Tx}
\end{table}
\fi
Let the shared key   $S_k ,k=0,1$ of user $k$ be a uniform binary random variable.
The cache encoding functions and the  transmission encoding function are denoted as
\begin{align*}
C_k(S_k, A,B) & = C_{k, S_k} (A,B) \text{ for } k=0,1, \\
E(A,B, D_0, D_1, S_0, S_1) & = T_{(D_0\oplus S_0, D_1 \oplus S_1)}(A,B).
\end{align*}
 User $k$ chooses $C_{k, S_k} (A,B)$ given in
Table~\ref{Table_cache_NK2} as the cache encoding function.
\ifarxiv
In the delivery phase, for given $(S_0,S_1)$ and $(D_0,D_1)$,
the server broadcasts both $(D_0 \oplus S_0,D_1
\oplus S_1)$ and $T_{(D_0 \oplus S_0,D_1 \oplus S_1)}(A,B)$.
\fi
 This choice of
transmission  satisfies the decodability condition due to the way we have
chosen the cache content and also since the chosen non-private scheme
satisfies the decodability condition for demands in \DRS. Further, the broadcast  transmission will
not reveal any information about the demand of one user to the other user since
one particular transmission $T_{(i,j)}(A,B)$ happens for all demand vectors
$(D_0,D_1)$, and also that $S_i$ acts as one time pad for $D_i$ for each $i=0,1$. Here,
all the transmissions consist of $4l$ bits (neglecting the 2 bits for $(D_0
\oplus S_0,D_1 \oplus S_1)$). Since $F=3l$, this scheme achieves a rate $R =
\frac{4}{3}$.
\end{example}

\ifarxiv
\begin{remark}
	It was shown in~\cite{AravindST19} using a dual construction
that if a memory-rate pair $(a,b)$
is achievable for $N=K=2$, then the pair $(b,a)$ is also achievable
if the cache contents are selected from one out of two
possible choices uniformly at random in the scheme for achieving $(a,b)$.
Furthermore, for $N=K=2$, the pair
$(4/3,1/3)$ is achievable using the MDS scheme given in~\cite{Wan19}. 
It can
also be observed that the pair $(4/3,1/3)$ is achievable by
broadcasting one bit  which implies that it preserves the information theoretic
privacy. However, it is not clear  whether we can achieve the pair $(1/3,4/3)$
directly using this MDS scheme and the dual construction in~\cite{AravindST19}
since in the MDS scheme the cache contents are not selected from one out of two
possible choices, as the dual construction requires. \end{remark}
\fi


{\it Proof of Theorem~\ref{Thm_genach}:}
Let us consider any $(N,NK,M,R)$ \DRS-non-private scheme. 
Let $C_k^{(np)}; k\in [NK]$ be the cache encoding functions, $E^{(np)}$ be the
broadcast encoding function, and $G_k^{(np)};k\in [NK]$ be the decoding
functions for the given $(N,NK,M,R)$ \DRS-non-private scheme.
We will now present a construction of an $(N,K,M,R)$-private scheme from
the given $(N,NK,M,R)$ \DRS-non-private scheme. 

\underline{Cache encoding:} For $k\in [K]$ and $S_k\in [N]$, the $k$-th user's cache
encoding function is given by
\begin{eqnarray}
&& C_k (S_k, \bar{W}) := C_{kN+S_k}^{(np)} (\bar{W}),
\label{eq:cachesch}
\end{eqnarray}
\ifarxiv
The $k^{\text{th}}$-user's cache encoding function is taken to be the same
as that of the $S_k$-th user in the $k$-th stack in the corresponding $(N,NK)$
caching problem.
\fi
The cache content is given by $Z_k=(C_k (S_k, \bar{W}), S_k)$.

\underline{Broadcast encoding:} To define the broadcast encoding, we need some
new notations and definitions.
Let $\Psi: [N]^N \rightarrow [N]^N$ denote the cyclic shift operator,
such that $\Psi (t_1,t_2,\cdots,t_N)=(t_N,t_1,\cdots,t_{N-1})$.
Let us denote a vector $\mathbb{I}:=(0,1,\cdots,N-1)$.
Let us also define 
\begin{align*}
\bar{S}\ominus \bar{D}:=(S_0\ominus D_0, S_1\ominus D_1, \cdots, S_{K-1}\ominus D_{K-1}),
\end{align*}
where $S_k\ominus D_k$ denotes the difference of $S_k$ and $D_k$ modulo $N$.
For a given $\bar{D}\in [N]^K$, we define an expanded demand vector
for the non-private problem as:
\begin{align*}
\bar{D}^{(np)}(\bar{D},\bar{S})=(\Psi^{S_0\ominus D_0}(\mathbb{I}),
\cdots,\Psi^{S_{K-1}\ominus D_{K-1}}(\mathbb{I})),
\end{align*}
 where $\Psi^i$ denotes
the $i$-times cyclic shift operator.

The broadcast encoding function for the $(N,K,M,R)$-private scheme is defined by
\begin{eqnarray}
&& E (\bar{W},\bar{D},\bar{S}) := E^{(np)}(\bar{W},\bar{D}^{(np)}(\bar{D},\bar{S})).
\label{eq:trencsch}
\end{eqnarray}

Let us denote $X_1=E (\bar{W},\bar{D},\bar{S})$.
The private scheme transmits the pair $X=(X_1, \bar{S}\ominus \bar{D})$.

\underline{Decoding:}
User $k$ uses the decoding function of the $(kN+S_k)$-th user in the
non-private scheme:
\ifarxiv
\begin{eqnarray}
&&\hspace*{-7mm}G_k(D_k, S_k, \bar{S}\ominus \bar{D}, X_1,Z_k)\nonumber\\ &&\hspace*{10mm}
=G_{kN+S_k}^{(np)}(\bar{D}^{(np)}(\bar{D},\bar{S}), X_1, Z_k)
\label{eq:decsch}
\end{eqnarray}
\fi
Here the decoder computes $\bar{D}^{(np)}(\bar{D},\bar{S})$ from
$\bar{S}\ominus \bar{D}$.

\underline{Proof of decodability:}
\ifarxiv
From \eqref{eq:decsch},\eqref{eq:cachesch}, and \eqref{eq:trencsch}, it is
clear that the decoder of the $k$-th user outputs the same file requested by the $S_k$-th virtual user of the $k$-th stack in the non-private scheme.
\fi
The index of the output file is the $(kN+S_k)$-th component in
$\bar{D}^{(np)}(\bar{D},\bar{S})$, i.e., $S_k\ominus (S_k\ominus D_k)=D_k$. 
Thus the $k$-th user recovers its desired file.

\underline{Proof of privacy:}

 The proof of privacy essentially follows from the fact that $S_i$ acts as
one time pad for $D_i$  which prevents any user $j \neq i$ getting any
information about $D_i$.
\ifarxiv
We now show that the derived $(N,K,M,R)$-private scheme satisfies the privacy condition~\eqref{Eq_instant_priv}. 
First we show that $I(\ND{k}; Z_k,D_k,\TRS | \bar{W}) =0$.
\begin{align}
& I(\ND{k}; Z_k,D_k,\TRS | \bar{W}) \nonumber \\
& = H(Z_k,D_k,\TRS | \bar{W}) - H(Z_k,D_k,\TRS | \bar{W}, \ND{k})  \nonumber \\
& =  H(S_k,D_k, \bar{S}\ominus \bar{D}, \bar{D}^{(np)}(\bar{D},\bar{S})  | \bar{W}) \nonumber \\
&\qquad - H(S_k,D_k, \bar{S}\ominus \bar{D}, \bar{D}^{(np)}(\bar{D},\bar{S}) | \bar{W},  \ND{k}) \label{Eq_privcy1}  \\
& =  H(S_k,D_k, \bar{S}\ominus \bar{D}  | \bar{W}) \nonumber \\
&\qquad - H(S_k,D_k, \bar{S}\ominus \bar{D} | \bar{W},  \ND{k}) \label{Eq_privcy1a}  \\
& =  H(S_k,D_k,\bar{S}\ominus \bar{D}) - H(S_k,D_k,\bar{S}\ominus \bar{D}  |  \ND{k}) \label{Eq_privcy2} \\
& =  H(S_k,D_k,\bar{S}\ominus \bar{D}) - H(S_k,D_k, \bar{S}\ominus \bar{D} ) \label{Eq_privcy3} \\
& =0. \label{Eq_privcy4}
\end{align}
Here~\eqref{Eq_privcy1} follows since $X=(X_1, \bar{S} \ominus \bar{D}), Z_k = (C_k(S_k,\bar{W}),S_k)$, and also 
due to~\eqref{eq:trencsch}. In~\eqref{Eq_privcy1a}, we used that $H(\bar{D}^{(np)}(\bar{D},\bar{S})| \bar{S}\ominus \bar{D}) =0$, and~\eqref{Eq_privcy2} follows since $(S_k,D_k, \bar{S}\ominus \bar{D}, \ND{k})$ is independent of $\bar{W}$.
We get~\eqref{Eq_privcy3} since $S_i \ominus D_i$ is independent of $D_i$ for all $i \in [K]$.
 Using the fact that demands and files are independent, 
we get 
the following from~\eqref{Eq_privcy4} 
\begin{align*}
I(\ND{k}; Z_k,D_k,\TRS , \bar{W}) & = I(\ND{k};  \bar{W}) + I(\ND{k}; Z_k,D_k,\TRS | \bar{W}) \\
& = 0.
\end{align*}
This shows the derived scheme satisfies the privacy condition 
$I(\ND{k}; Z_k,D_k,\TRS) =0$.  
\fi

\ifarxiv
The size of the cache in the $(N,K,M,R)$-private scheme differs only by the size of the shared key from the $(N,NK,M,R)$ \DRS-non-private scheme. For large enough file size $2^F$, this difference is negligible. We can also observe that the rate of transmission in $(N,K,M,R)$-private scheme is the same as  that of the $(N,NK,M,R)$ \DRS-non-private scheme.
This proves Theorem~\ref{Thm_genach}.
\fi
\hfill{\rule{2.1mm}{2.1mm}}

\ifarxiv
We have the following corollary.
\fi

\begin{corollary}
		\label{Cor_Type}
 If there exists an $(N,NK,M,R)$ $\cD_{(K,K,\ldots,K)}$-non-private scheme, then
there exists an $(N,K,M,R)$-private scheme.
\end{corollary}
\ifarxiv
\begin{proof}
	As mentioned before, we have $\DRS \subseteq \cD_{(K,K,\ldots,K)}$. So an $(N,NK,M,R)$ $\cD_{(K,K,\ldots,K)}$-non-private scheme is also an $(N,NK,M,R)$ $\DRS$-non-private scheme. Then, the corollary follows from Theorem~\ref{Thm_genach}.
\end{proof}
\fi

The converse proof of Theorem~\ref{Thm_N2K2_region} uses the following lemma on some conditional distributions. 
\ifarxiv
\begin{lemma}
	\label{Lem_eqiuv_distrbn2}
		Let $\tilde{i} = (i+1) \mod 2$ for $i=0,1$. Then  any demand private scheme for $N=K=2$ satisfies the following  for $j =0,1$:
	\begin{align}
	( X, Z_1, W_{j}|D_1=j) & \sim ( X, Z_1, W_{j}|D_0=i, D_1=j) \notag\\
	&  \sim  (X, Z_1, W_{j}|D_0 = \tilde{i}, D_1=j) \label{Eq_priv_cond3}
	\end{align}
	and
	\begin{align} 
	 (X, Z_0,  W_{j}|D_0=j)  & \sim (X, Z_0,  W_{j}|D_0=j, D_1=i)  \notag\\
	& \sim (X,  Z_0, W_{j} | D_0 = j, D_1=\tilde{i}). \label{Eq_priv_cond4}
	\end{align}
\end{lemma}
\fi
\ifarxiv
\begin{proof}
		We prove~\eqref{Eq_priv_cond3} for $D_1 =0$. Other cases follow similarly.
		Any $(N,K,M,R)$-private scheme satisfies that $I(D_0;Z_1,D_1,X) =0$. 
	Since $H(W_{D_1}|X,Z_1,D_1) =0$,  we have that $I(D_0;Z_1,D_1,X, W_{D_1}) =0$. Then it follows that
	\begin{align*}
	& Pr(D_0=0  | X=x, Z_1 = z', W_0=w_0,D_1=0)\\
	&\quad =Pr(D_0=1| X=x, Z_1 = z', W_0=w_0,D_1=0).
	\end{align*}
	Multiplying both sides by \\$Pr(X=x, Z_1 = z', W_0=w_0 | D_1=0)$ gives
	\begin{align*}
	&Pr(D_0=0, X=x, Z_1 = z',W_0=w_0 | D_1=0) \\
	&=Pr(D_0=1, X=x, Z_1 = z', W_0=w_0| D_1=0).
	\end{align*}
	Then it follows that 
	\begin{align*}
	& Pr(D_0=0| D_1=0) \times \\
	&\quad \quad   Pr(X=x, Z_1 = z', W_0=w_0| D_0=0,D_1=0) \\
	&= Pr(D_0=1| D_1=0)\times  \\
	&\quad \quad  Pr(X=x, Z_1 = z',W_0=w_0,|  D_0=1, D_1=0).
	\end{align*}
	Since the demands are equally likely and they are independent of each other, we get
		\begin{align}
 & Pr(X=x, Z_1 = z', W_0=w_0| D_0=0,D_1=0)\notag \\
	&=Pr(X=x, Z_1 = z',W_0=w_0,|  D_0=1, D_1=0). \label{eq:lemp1}
	\end{align}
	Further, we also have
	\begin{align}
	& Pr(X=x, Z_1 = z', W_0=w_0| D_1=0)\notag \\
	&= \sum_{i=0}^{1}Pr(D_0=i) \times \notag\\
	&\hspace{4mm} Pr(X=x, Z_1 = z',W_0=w_0,|  D_0=i, D_1=0). \label{eq:lemp2}
	\end{align}
		Eq.~\eqref{eq:lemp1} and \eqref{eq:lemp2} together prove~\eqref{Eq_priv_cond3} for $D_1 =0$.
\end{proof}
\fi
\begin{figure}
	\centering
\ifarxiv
	\includegraphics[scale=0.48]{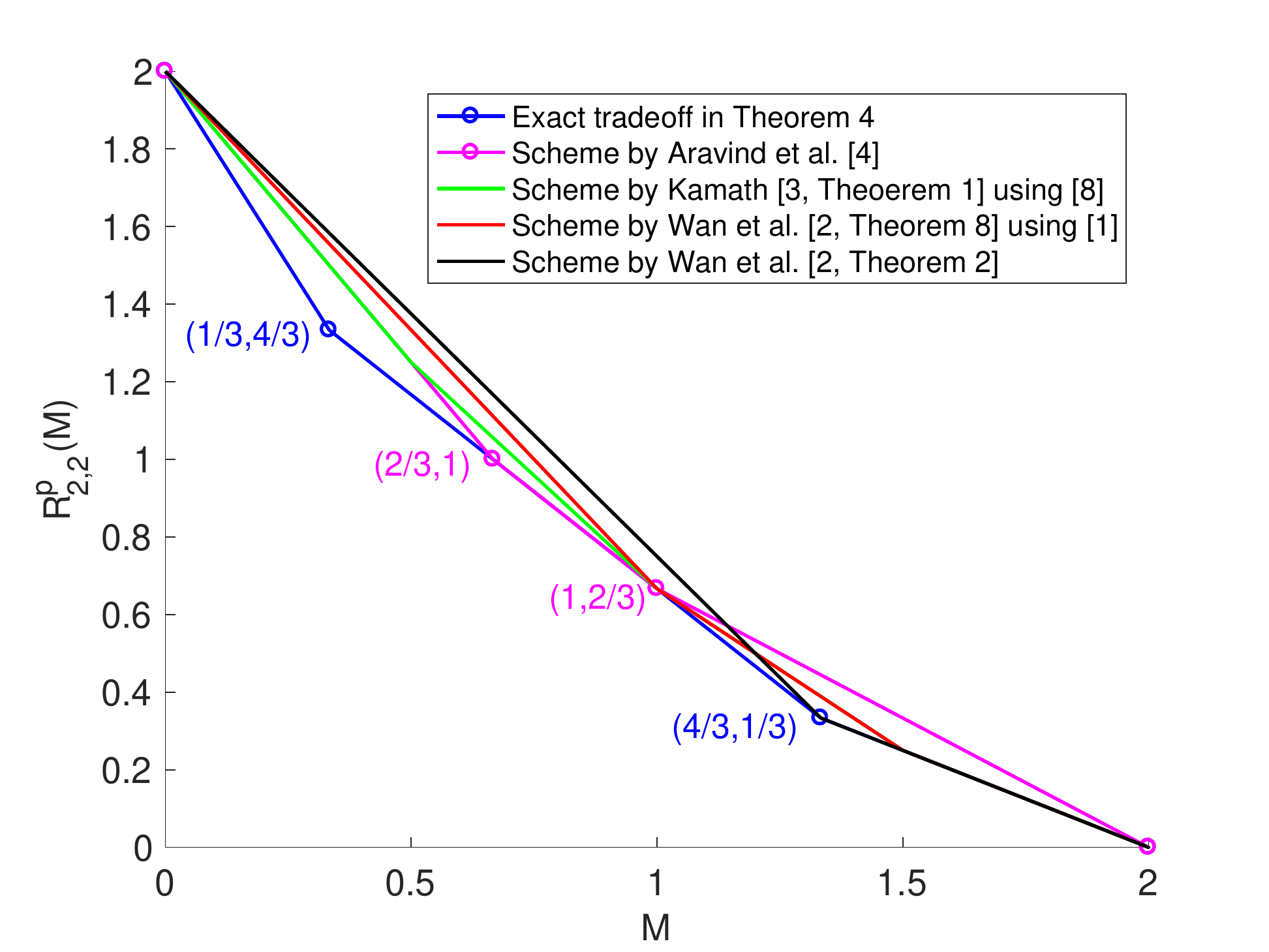}
\fi
	\caption{Comparison of known private schemes for $N=K=2$.}
	\label{Fig_region_N2K2}
\end{figure}

We present the optimal memory-rate region with demand privacy for $N=K=2$ in
Theorem~\ref{Thm_N2K2_region}. In Fig.~\ref{Fig_region_N2K2}, we plot the optimal trade-off for $N=K=2 $ along with the known achievable memory-rate pairs using different schemes in the literature.
\ifarxiv
Other points known using other schemes are also shown.  
\fi

 \begin{theorem}
 	\label{Thm_N2K2_region}
 	Any memory-rate pair $(M,R)$ is achievable with demand privacy for $N=K=2$ if and only if 
 	\begin{align}
 	2M + R \geq 2, \quad 3M+3R \geq 5, \quad M+2R \geq 2. \label{Eq_N2K2_region}
 	\end{align}
 \end{theorem}
\begin{proof}
\ifarxiv
	We first show the achievability of the region given in~\eqref{Eq_N2K2_region}.
\fi
	It was shown in~\cite[Proposition~7]{Tian2018} that the region given
by~\eqref{Eq_N2K2_region} is an achievable rate region for Type (2,2) in
$N=2,K=4$  coded caching problem. Then the achievability under demand-privacy
for $N=K=2$ follows from
Corollary~\ref{Cor_Type}. 
\ifarxiv
We note that to show the achievability of the region
given by~\eqref{Eq_N2K2_region}  for Type (2,2) in $N=2,K=4$  coded caching
problem, in particular it was shown that  the memory-rate pairs
$(\frac{1}{3},\frac{4}{3})$ and $(\frac{4}{3}, \frac{1}{3})$ are
achievable~\cite{Tian2018}. The scheme that achieves the memory-rate pair
$(\frac{1}{3},\frac{4}{3})$ was used in Example~\ref{Ex_simple} for achieving
the same point for the demand-private scheme for $N=K=2$. The point $(\frac{4}{3}, \frac{1}{3})$ is also achievable by MDS scheme in~\cite{Wan19}.
\fi

	To show the converse, we only need to prove that  any $(M,R)$ pair
satisfies $3M+3R \geq 5$. The other two inequalities in~\eqref{Eq_N2K2_region} 
are also necessary under no privacy requirement. So, those  hold  under privacy requirement as well. 
 We note that the bound  $3M+3R \geq 5$ is given for 2 files and 3 users  under no privacy requirement in~\cite[Proposition~5]{Tian2018}. We obtain this bound for 2 users 2 files with privacy constraint, crucially using Lemma~\ref{Lem_eqiuv_distrbn2}. To this end, we first lower bound  $H(Z_0,X|D_0=0)+ H(Z_1,X|D_1=0) + H(Z_1,X|D_0=1,D_1=0)$ by $5F$ and then upper bound the same quantity by  $3MF+3RF$ which proves the bound $3M+3R \geq 5$.

\ifarxiv
We first  lower bound $H(Z_0,X|D_0=0)+ H(Z_1,X|D_1=0) $ as follows:
\begin{align}
	&H(Z_0,X|D_0=0)+ H(Z_1,X|D_1=0)  \nonumber \\
	&=H(Z_0,W_0,X|D_0=0)+ H(Z_1,W_0,X|D_1=0) \label{Eq_N2K2_conv1} \\
	&=H(Z_0,W_0,X|D_0=0,D_1=0)  \nonumber \\ &\quad
          + H(Z_1,W_0,X|D_0=0,D_1=0) \label{Eq_N2K2_conv1a} \\
	&\geq H(Z_0,Z_1|W_0,X,D_0=0,D_1=0)  \nonumber \\ &\quad \quad 
          + 2H(W_0,X|D_0=0,D_1=0) \nonumber \\
	&=H(Z_0,Z_1,W_0,X|D_0=0,D_1=0) \nonumber \\
	&\quad \quad + H(W_0,X|D_0=0,D_1=0), \label{d_eq1}
\end{align}
where~\eqref{Eq_N2K2_conv1} follows from the decodability condition, and~\eqref{Eq_N2K2_conv1a} follows from Lemma~\ref{Lem_eqiuv_distrbn2}.
\fi
\ifarxiv
Now we find a lower bound on $H(Z_0,Z_1,W_0,X|D_0=0,D_1=0) + H(Z_1,X|D_0=1,D_1=0)$.
\begin{align}
	&H(Z_0,Z_1,W_0,X|D_0=0,D_1=0) \nonumber\\
	&\quad \quad + H(X,Z_1|D_0=1,D_1=0) \nonumber\\
	&= H(Z_0,Z_1,W_0,X|D_0=0,D_1=0) \nonumber \\
	&\quad \quad + H(X,Z_1,W_0|D_0=1,D_1=0) \label{Eq_N2K2_conv2}\\
	&\geq H(Z_0,Z_1,W_0|D_0=0,D_1=0) \nonumber\\
	&\quad \quad + H(X,Z_1,W_0|D_0=1,D_1=0) \nonumber\\
	&= H(Z_0,Z_1,W_0|D_0=1,D_1=0) \nonumber\\
	&\quad \quad + H(X,Z_1,W_0|D_0=1,D_1=0) \label{Eq_N2K2_conv3}\\
	&\geq H(Z_0,X|Z_1,W_0,D_0=1,D_1=0) \nonumber\\
	&\quad \quad + 2H(Z_1,W_0|D_0=1,D_1=0) \nonumber  \\
	&= H(Z_0,X, W_1|Z_1,W_0,D_0=1,D_1=0) \nonumber\\
	&\quad \quad + 2H(Z_1,W_0|D_0=1,D_1=0) \label{Eq_N2K2_conv4}\\
	&\geq H(W_1,Z_1,W_0|D_0=1,D_1=0) \nonumber\\
	&\quad \quad + H(Z_1,W_0|D_0=1,D_1=0) \nonumber\\
	&\geq H(W_1,W_0) + H(Z_1,W_0|D_0=1,D_1=0). \label{d_eq2}
\end{align}
\fi

\ifarxiv
From~\eqref{d_eq1} and~\eqref{d_eq2}, we  obtain
\begin{align}
	&H(Z_0,X|D_0=0)+ H(Z_1,X|D_1=0) \nonumber \\
	&\quad \quad + H(Z_1,X|D_0=1,D_1=0) \nonumber \\
	&\geq H(W_1,W_0) + H(Z_1,W_0|D_0=1,D_1=0) \nonumber \\
	&\quad \quad +H(W_0,X|D_0=0,D_1=0) \nonumber \\
	& = H(W_1,W_0) + H(Z_1,W_0|D_0=0,D_1=1) \nonumber \\
	&\quad \quad + H(W_0,X|D_0=0,D_1=1) \label{Eq_N2K2_conv5} \\
	& \geq  H(W_1,W_0) + H(Z_1,X,W_1|W_0,D_0=0,D_1=1) \nonumber \\
	&\quad \quad + 2H(W_0|D_0=0,D_1=1) \label{Eq_N2K2_conv6} \\
	& \geq H(W_1,W_0) + H(W_1|W_0,D_0=0,D_1=1) \nonumber \\
	&\quad \quad + 2H(W_0|D_0=0,D_1=1) \nonumber\\
	& = 5F,\label{d_eq3}
\end{align}
where in~\eqref{Eq_N2K2_conv5}, the equality in the third term follows from Lemma~\ref{Lem_eqiuv_distrbn2} and the equality in the second term follows since demands are independent of caches and files. Further,
\eqref{Eq_N2K2_conv6} follows because of the  decodability condition and also due to the fact that conditioning reduces entropy. 
\fi
\ifarxiv
We now show an upper bound of the same quantity:
\begin{align}
	&H(Z_0,X|D_0=0)+ H(Z_1,X|D_1=0) \nonumber \\
	&\quad \quad + H(Z_1,X|D_0=1,D_1=0) \nonumber \\
	&\leq H(Z_0|D_0=0) + H(X|D_0=0) + H(X|D_1=0) \nonumber \\
	&\quad \quad + H(Z_1|D_1=0) + H(Z_1|D_0=1, D_1=0) \nonumber\\
	&\qquad \quad  + H(X|D_0=1,D_1=0)\nonumber\\
	&\leq 3MF+3RF, \label{d_eq4}
\end{align}
where in the last inequality we used the fact that  the cache contents  are independent of the demands and also that the transmission is independent of demands due to the privacy condition. From~\eqref{d_eq3} and~\eqref{d_eq4}, we get $3M+3R\geq5$. This completes the proof of Theorem~\ref{Thm_N2K2_region}.
\fi
\end{proof}

%

\ifarxiv
 \section{Conclusion}
\label{sec_concl}
The problem of demand-private coded caching opens new avenues of research.
In this paper, we considered the demand-private
coded caching problem for noiseless broadcast network under
information-theoretic privacy. To characterize the exact trade-off for $N=K=2$,
we have provided a converse bound that accounts for the privacy constraints. To
the best of our knowledge, this converse bound is the first of its kind. It is
well known from the non-private coded caching problem that obtaining the exact
trade-off is in general difficult. So, we expect the same would be true for
demand-private coded caching. Since the exact characterization is known to be
difficult, the order optimality of the achievable schemes are
investigated in non-private schemes.  Theorem~\ref{Thm_order} completes the
order optimality result from \cite{Kamath19} for information
theoretically demand-private coded
caching. The order optimality is shown by comparing the achievable rates under
privacy with the lower bounds on the rates under no privacy. This also
highlights the fact that the optimal rates with and without demand privacy are
always within a constant factor.
\fi



 \bibliographystyle{IEEEtran}
 \bibliography{Bibliography.bib}





\end{document}